%% file: main_ipco.tex
\documentclass[runningheads]{llncs}
\usepackage[T1]{fontenc}
\usepackage{amsmath, amsfonts, amssymb}
\usepackage{graphicx, xcolor}
\usepackage{graphics} 
\usepackage[vlined,ruled,linesnumbered]{algorithm2e}
    \SetEndCharOfAlgoLine{}
    \SetKwInput{KwInput}{Input}
    \SetKwInOut{KwInit}{Init}
    \SetKw{Continue}{continue to next iteration}
    \SetKw{Break}{break out of iteration}
    \RestyleAlgo{boxruled}
    \let\oldnl\nl
    \newcommand{\nonl}{\renewcommand{\nl}{\let\nl\oldnl}}
\usepackage{setspace}
\usepackage{ifthen}
\usepackage{tabularx, colortbl} 
\usepackage[inline]{enumitem}
\usepackage[hidelinks]{hyperref}
\usepackage{bold-extra}

\renewcommand{\d}[1]{d_{#1}}
\newcommand{\f}[2]{\ifthenelse{\equal{#2}{}}{f_{#1}}{f_{#1}(#2)}}
\newcommand{\x}[2]{x_{#1,#2}}

\newcommand{\y}[1]{y_{#1}}

\renewcommand{\a}[2]{a_{#1,#2}}
\newcommand{\lp}{\textsc{lp(\ref{LP:blocks}-\ref{LP:y-nonneg})}}
\newcommand{\constraints}{(\ref{LP:blocks}-\ref{LP:y-nonneg})}

\newcommand{\cms}{\textsc{cms}}
\newcommand{\numcms}{Numerical-\textsc{cms}}

\newcommand{\opt}{\textsc{opt}}
\newcommand{\eps}{\varepsilon}
\newcommand{\set}[1]{\left\{ #1 \right\}}

\newcommand{\floor}[1]{\left\lfloor #1 \right\rfloor}

\newenvironment{heuristic}[1][]{\refstepcounter{algocf}\par\medskip
   \noindent \textbf{Algorithm~\arabic{algocf}.}\ \rmfamily}{\par\medskip}

\newcommand{\junk}[1]{}
\newcounter{tempcounter}

\begin{document}

\title{Scheduling Splittable Jobs on\\Configurable Machines}
\titlerunning{Configurable Machines}

\author{Matthew Casey \and Rajmohan Rajaraman \and David Stalfa \and Cheng Tan}

\authorrunning{M. Casey et al.}

\institute{Northeastern University, Boston MA 02115, USA \\ \email{\{casey.ma, r.rajaraman, stalfa.d, c.tan\}@northeastern.edu}}


\maketitle

\begin{abstract}
Motivated by deep neural network applications, we study the problem of scheduling splittable jobs (e.g., neural network inference tasks) on configurable machines (e.g., multi-instance GPUs). We are given $n$ jobs and a set $C$ of configurations (e.g, representing ways to configure a GPU) consisting of multisets of blocks (e.g., representing GPU instances). A schedule consists of a set of machines, each assigned some configuration in $C$ with each block in the configuration assigned to process one job. The amount of a job's demand that is satisfied by a given block is an arbitrary function of the job and block. The objective is to satisfy all demands on as few machines as possible. We provide a tight logarithmic approximation algorithm for this problem in the general setting, an asymptotic $(2 + \eps)$-approximation with $O(1)$ input configurations for arbitrary $\eps > 0$, and a polynomial time approximation scheme when both the number and size of configurations are $O(1)$. 

\junk{
We also provide constant approximation bounds in the setting where machines have a fixed size $k$ and any integer partition of $k$ is a valid configuration.}

\end{abstract}

\input{intro}

\input{LP_greedy}

\input{LP_pseudo}

\input{comb_const_conf_size}

\bibliographystyle{splncs04}
\bibliography{refs}

\appendix
\input{appendix}

\end{document}

%% file: intro.tex
\section{Introduction}
\label{sec:intro} 
Deep neural network models, especially LLMs, are extremely resource intensive and require careful allocation of resources to maximize throughput at the time of inference.  Each DNN inference job either consists of a sequence of inference queries, or is a long-running request needing a certain throughput of inference queries. These jobs are typically assigned multiple GPUs, each running the same underlying model and processing inference query streams.  The performance of a DNN model (measured by the throughput they achieve or the latency they provide for inference tasks) does not always vary linearly with the resources provided; so, allocating a full GPU instance to a given DNN inference job may be wasteful in some scenarios.  Modern GPUs (e.g., Nvidia's A100) include a feature called Multi-Instance GPU, which enables a GPU to be configured into smaller isolated instances, each with their own processors, memory, and L2 cache. Recent work~\cite{tan2021serving} has argued that this configurability can yield much more cost-effective execution of DNN inference jobs by partitioning individual GPUs into smaller instances and allocating the DNN inference jobs to instances of appropriate size.

In this work, we initiate a systematic study of scheduling splittable jobs in configurable machines. We call this problem \emph{Configurable Machine Scheduling} or \cms.  We consider machines that can be configured into smaller instances, which we call \emph{blocks}, in multiple ways, each of which is referred to as a \emph{configuration}.  We consider jobs, each with a certain demand that needs to be satisfied by allocating blocks.  Each job has a table that specifies how much demand can be satisfied by a given block type.  The desired output of the problem is the number of machines of each configuration type and the number of blocks of each block type to allocate for each job, subject to two constraints: (i) the blocks allocated for each job ensure that the demand of the job is satisfied, and (ii) the blocks allocated for each block type match the number of blocks in the machine configurations.  We focus on the goal of minimizing the total number of machines.

\bigskip
\input{problem}
\bigskip

\noindent\textbf{Our results}
\smallskip

\noindent Our \cms\ problem formulation yields a rich landscape of optimization problems, which vary depending on the properties of block types, configurations, and the job demand tables.  In this paper, we focus on 
the general \cms\ problem and two special cases where the number of configurations is bounded.  We obtain near-tight approximation results (see Table~\ref{tab:bounds}) for the associated problems.

\junk{
\noindent We introduce two variants of configurable machine scheduling: (i) a general \emph{combinatorial} variant \cms\ where each configuration is an arbitrary multiset of block types, and (ii) a \emph{numerical} variant \numcms\ where each block type has an integer size and the blocks in each configuration have total size equal to a given fixed integer, the size of each machine. The combinatorial \cms\ generalizes multiset multicover~\cite{rajagopalan+v:cover}, while \numcms\ generalizes bin-packing~\cite{XXX}. The prevalence of interesting special cases of \cms\ yields a rich landscape of optimization problems, which we begin to explore.
We obtain near-tight approximation results (see Table~\ref{tab:bounds}) for 
the associated problems.}  
\smallskip

\noindent {\bfseries General \cms\ (Section~\ref{sec:LP+greedy}).}\ Using a reduction from minimum multiset multicover~\cite{rajagopalan+v:cover}, we observe that \cms\ is hard to approximate to better than a factor of $\Omega(\log nk)$, where $n$ is the number of jobs and $k$ the number of blocks.  We then present an $O(\log (cnk))$-approximation algorithm, where $n$ is the number of jobs, $k$ the number of blocks, and $c$ is the size of the largest configuration. Our algorithm constructs a schedule by greedily selecting the highest throughput configuration on the basis of a linear programming relaxation.
\smallskip

\noindent {\bf \cms\ with $O(1)$ configurations (Section~\ref{sec:cms_constant}).}\ Using a reduction from Partition, we observe that \cms, \emph{even with one configuration and two jobs}, is hard to approximate to better than a factor of 2. We present an algorithm that, for any instance of \cms\ with $O(1)$ configurations $C$ and arbitrary $\eps > 0$, uses at most $(2 + \eps)\opt + |C|$ machines where $\opt$ is the number of machines needed in the optimal solution. We also show that our algorithm always achieves a $3 + \eps$ approximation. Our algorithm builds on the seminal LP rounding technique of~\cite{LenstraShmoysTardos} and exploits the structure of extreme-point solutions to iteratively and carefully round the LP variables.   
\smallskip

\noindent {\bf \cms\ with $O(1)$ configurations of $O(1)$ size (Section~\ref{sec:ptas_constantk}).}\ We next consider combinatorial \cms\ with a constant number of configurations, each of constant size (i.e., having a constant number of blocks).  We show that the problem is solvable in pseudo-polynomial time; our main result here is a PTAS based on rounding a novel LP relaxation for the problem.  
\junk{\smallskip 

\noindent {\bf \numcms\ with all configurations allowed (Section~\ref{sec:numerical}).}\ Finally, we consider \numcms, which is NP-hard even when all configurations are allowed.  We present a $(2 + \eps)$-approximation algorithm, for any $\eps > 0$, that combines a greedy algorithm with a PTAS for unbounded min-knapsack.}



\begin{table}
\begin{center}
\begin{tabular}{|cccc|}
    \hline
    \rowcolor{lightgray} &&& \\[-8pt]
    \rowcolor{lightgray}Problem&Algorithm&\ Approximation\ &Hardness\\[2pt]
    \hline &&& \\[-10pt] 
    \hline &&&\\[-6pt]
    \cms & LP + Greedy & $O(\log cnk)$ & $\Omega(\log nk)$ \\[6pt]
    \hline  &&&\\[-8pt]
    {\begin{tabular}{@{}c@{}}\cms\\ \ $O(1)$ configurations \ \end{tabular}} & {\begin{tabular}{@{}c@{}}Extreme-Point\\ \ LP Rounding \ \end{tabular}} &  {\begin{tabular}{@{}c@{}} \ $(2+\varepsilon)\opt + |C|$ \ \\ $3+\varepsilon$ \end{tabular}} & $2$ \\[10pt]
    \hline  &&&\\[-8pt]
    {\begin{tabular}{@{}c@{}}\cms\\\ $O(1)$ configurations \\of $O(1)$ size\end{tabular}} & Small/Large Job LP &  $1 + \varepsilon$ & ? \\[15pt]
    \hline
\end{tabular}
\end{center}
\caption{Results for Configurable Machine Scheduling. $n$ is the number of jobs, $k$ is the number of block types, and $c = \max_{\sigma \in C}\{|\sigma|\}$ is the maximum size of any configuration.}
\label{tab:bounds}
\vspace{-7mm}
\end{table}%

\noindent\textbf{Related work}
\smallskip

\noindent Configurable machine scheduling has connections to many well-studied problems in combinatorial optimization, including bin-packing, knapsack, multiset multicover, and max-min fair allocation.  The general combinatorial \cms\ problem generalizes the multiset multicover problem~\cite{korte+v:book,hua+wyl:multicover,rajagopalan+v:cover}, for which the best approximation factor achievable in polynomial time is $O(\log m)$ where $m$ is the sum of the sizes of the multisets~\cite{rajagopalan+v:cover,Vazirani}.  The hardness of approximating the problem to within an $O(\log n)$ factor follows from the result for set cover~\cite{dinur+s:repetition}.

As we note above, combinatorial \cms\ is NP-complete even for the case of one configuration and two jobs.  The single configuration version can be viewed as a fair allocation problem with each block representing an item and each job representing a player that has a value for each item (given by the demand table) and a desired total demand.  The objective then is to minimize the maximum number of copies we need of each block so that they can be distributed among the players satisfying their demands.  In contrast, the Santa Claus problem in fair allocation~\cite{bansal+s:santa_claus} (also studied under a different name in algorithmic game theory~\cite{lipton+mms:fair}) aims to maximize the minimum demand that can be satisfied with the available set of blocks.  The best known approximation algorithm for the Santa Claus problem  is a quasi-polynomial time $O(n^{\eps})$-approximation,
where $\eps = O(\log\log n/ \log n)$~\cite{chakrabarty+ck:fair}, though $O(1)$ approximations are known for special cases (e.g., see~\cite{cheng+m:fair}).    

\junk{The \cms\ problem combines aspects of bin-packing, knapsack, and matching: as in knapsack, blocks are chosen based on the demand they satisfy; jobs are matched to blocks, which are matched to configurations; and \junk{(in the numerical case)} machines are treated as (abstract) bins into which blocks are packed.  

Nevertheless, the PTAS for knapsack and asymptotic PTAS for bin-packing do not extend to \numcms.  Our algorithm for \numcms\ 
achieves a $(2 + \eps)$-approximation by combining a greedy technique with an algorithm for unbounded knapsack~\cite{JianZhao.UnboundedKnapsack.2019}.  If further the input instance has only a constant number of job types (as is likely in practice), the techniques of Goemans-Rothvoss~\cite{GoemansRothvoss.BinPacking.13} can be used to yield a polynomial-time optimal solution.}
\bigskip

\noindent \textbf{Discussion and Open Problems} 
\smallskip

\noindent Our study has focused on a \emph{combinatorial} version of \cms\
in which each machine can be configured as a collection of abstract blocks.  It is also natural to consider a \emph{numerical} version of \cms\ in which each block type is an item of a certain size, and each configuration has a certain capacity and can only fit blocks whose sizes add up exactly to its capacity.  The approximation ratios established for \cms\ apply to numerical \cms\ as well; however it is not certain that there is also a logarithmic hardness for numerical \cms.  Thus, an intriguing open problem is whether numerical \cms\ admits an approximation factor significantly better than the logarithmic factor established in Section~\ref{sec:LP+greedy}. Also of interest is a numerical \cms\ variant where all capacity-bounded configurations are allowed, for which we believe techniques from unbounded knapsack and polytope structure results from bin-packing would be useful~\cite{JianZhao.UnboundedKnapsack.2019,GoemansRothvoss.BinPacking.13}.  

\junk{Despite the similarity to knapsack and bin-packing, the PTAS for knapsack and asymptotic PTAS for bin-packing do not extend to \numcms.  Our algorithm for \numcms\ 
achieves a $(2 + \eps)$-approximation by combining a greedy technique with an algorithm for unbounded knapsack~\cite{JianZhao.UnboundedKnapsack.2019}.  If further the input instance has only a constant number of job types (as is likely in practice), the techniques of Goemans-Rothvoss~\cite{GoemansRothvoss.BinPacking.13} can be used to yield a polynomial-time optimal solution.}

Our results indicate several directions for future research. One open problem is to devise approximation algorithms that leverage structure in the set of available configurations. In practice, the configuration sets associated with multi-instancing GPUs might not be arbitrary sets, e.g. the blocks of Nvidia's A100 GPU are structured as a tree and every valid configuration is a set of blocks with no ancestor-descendant relations~\cite{tan2021serving}. Showing improved bounds for such cases seems to be a challenging, but potentially fruitful area of research.


Another open problem lies in shrinking the gap between our upper and lower bounds. The hard instances for \cms\ with $O(1)$ configurations and \numcms\ have constant size solutions, showing e.g. that it is NP-hard to distinguish a problem with solution size 1 from one with solution size 2. These lower bounds are sufficient to show hardness of approximation, but do not rule out the possibility of asymptotic PTAS (even additive constant approximations). Furthermore, we have not been able to show any hardness for \cms\ with $O(1)$ configurations of $O(1)$ size, doing so is an important and interesting open problem.

Finally, our focus has been on the objective of minimizing the number of machines, which aims to meet all demands using minimum resources.  Our results can be extended to minimizing makespan, given a fixed number of machines.  However, approximations for other objectives such as completion time\junk{(for which we have preliminary results in Appendix~\ref{sec:numerical})} or flow time, in both offline and online settings, are important directions for further research.



%% file: problem.tex

\noindent \textbf{Configurable Machine Scheduling (\cms)}

\smallskip
\noindent We are given a set $J$ of $n$ \textit{jobs} and a set $B = \{1,2,\ldots,k\}$ of $k$ \textit{block types}. Each job $j$ has an associated \textit{demand} $\d j$ and \textit{demand table} $\f j{}$. For each element $i \in B$, the function $\f ji$ indicates how many units of $j$'s demand is satisfied by a block of type $i$. (We assume that $\max_i \{\f ji\} \le \d j$ and that $\min_{i:\f ji \neq 0}\{\f ji\} = 1$. The former can be achieved by reducing large values, and the latter by scaling all table values and demands, neither of which affects the optimal solution.)

A \textit{configuration} $\sigma$ is a multiset of blocks in $B$. A \textit{machine} $\mu$ is a mapping from the blocks of some configuration to jobs, and a schedule $S$ consists of a set of multiplicity-machine pairs $(a,\mu)$. 
For each job, the sum of demands satisfied by all blocks assigned to the job must be at least the job's demand, i.e. for each job $j$, $\sum_{(a,\mu) \in S} a \cdot  \sum_{i:\mu(i)=j} \f ji \ge \d j$. Our objective is to construct a schedule that minimizes the number of machines (i.e. minimizes $\sum_{(a,\mu) \in S} a$). 

\junk{We distinguish between two models: the combinatorial and the numerical. The models provide different restrictions on how machines can be configured. }

A problem instance is specified as a triple $(C,f,d)$ where $C$ is a set of allowable configurations, where each configuration $\sigma$ is multiset of elements in $B$. $f$ is an $n\times k$ matrix specifying the demand table for each job, and $d$ is the vector of their demands. 

%% file: LP_greedy.tex
\section{Logarithmic approximation for \cms}
\label{sec:LP+greedy}

In this section, we consider the most general model of \cms\ with an arbitrary configuration set $C$ over $k$ blocks, and $n$ jobs with demand functions $f$ and demands $d$.  The main result of this section is an $O(\log(\max_{\sigma \in C} \{|\sigma|\} \cdot n \cdot k)$-approximation algorithm for \cms\ given by Algorithm~\ref{alg:log-approx}. 

The following lemma presents an approximation-preserving reduction from multiset multicover to \cms, which implies that no polynomial time algorithm can achieve an approximation ratio better than $\Omega(\log nk)$ (assuming $\textsc{p} \neq \textsc{np}$). The lemma also implies that an improvement to our approximation ratio would yield an improvement to the best known approximation for multiset multicover. (For the proof of Lemma~\ref{lem:cms_hardness}, see Appendix~\ref{sec:cms_appendix}).
\begin{lemma}
    There is an approximation-preserving reduction from the multiset multicover problem to \cms.
\label{lem:cms_hardness}
\end{lemma}


\vspace{-0.3in}
\begin{algorithm}
\label{alg:log-approx}
    \textit{Formulate and Solve a Linear Relaxation (Constraints~1-4)} \hspace{\textwidth}
    Round variables down if their fractional component is less than $(1/2k)$\;
    \nonl\;
    
    \vspace{-2mm}
    \textit{Solve Problem over the Integer Components of Variables (Algorithm~\ref{alg:multisetmulticover_reduction})} 
    Solve Multiset Multicover problem defined by integer components of optimal solution to construct a partial schedule $S_1$\; \nonl\;

    \vspace{-2mm}
    \textit{Greedily Round the Fractional Components of Variables (Algorithm~\ref{alg:greedy})} \hspace{\textwidth} 
    Construct a partial schedule $S_2$ to satisfy any remaining demand by greedily configuring each machine to maximize throughput \; \nonl \;

    \vspace{-2mm}
    \textit{Output the schedule formed by the additive union $(S_1 \oplus S_1) \oplus (S_2 \oplus S_2)$} 
\caption{Logarithmic Approximation for \cms}
\end{algorithm}
\vspace{-0.2in}

The first step of Algorithm~\ref{alg:log-approx} consists in defining and solving the linear program \constraints, which minimizes $\sum_\sigma \y \sigma$ subject to:

\vspace{-0.2in}
\begin{align}
    &\textstyle\sum_j \x ij \le \sum_{\sigma \in C} \y \sigma \cdot \a \sigma i &&\forall \ \text{block types}\ i \in B \label{LP:blocks} \\
    &\textstyle\sum_i \f ji \cdot \x ij \ge \d j &&\forall \ \text{jobs}\ j \label{LP:demand} \\
    &\x ij \ge 0 &&\forall \ \text{block types}\  i \in B \ \text{and jobs}\ j \label{LP:x-nonneg} \\
    &\y \sigma \ge 0 &&\forall \ \text{configurations}\  \sigma \in C \label{LP:y-nonneg}
\end{align}

\textbf{Terms.} 
Each variable $\x ij$ indicates the number of blocks of type $i$ that are assigned to execute job $j$. Each variable $\y \sigma$ indicates the number of machines that use configuration $\sigma$. 
The term $\a \sigma i$ is the (constant) number of blocks of type $i$ in configuration $\sigma$.

\textbf{Constraints.}
Constraint~\ref{LP:blocks} ensures a schedule cannot use more blocks of a given type than appear across all allocated machines. Constraint~\ref{LP:demand} states that the total number of blocks executing a job must be sufficient to satisfy its demand. 

\medskip
\noindent Let $(x^*, y^*)$ be an optimal solution to \constraints. For the second step of Algorithm~\ref{alg:log-approx}, we separate the integer from the fractional components of the $x$-variables.  We define $\bar x_{i,j} = \floor{x^*_{i,j}}$. Let $z^*_{i,j} = x^*_{i,j} - \floor{x^*_{i,j}}$. We define $\hat x_{i,j} = 0$ if either (i) $z^*_{i,j} < \frac{1}{2k}$ or (ii) $\f ji \cdot z^*_{i,j} < \max_{i'}\{\f j{i'} \cdot z^*_{i',j}\}/k$, otherwise  $\hat x_{i,j} = 2z^*_{ij}$. The second step of Algorithm~\ref{alg:log-approx} then uses Algorithm~\ref{alg:multisetmulticover_reduction} to provide a schedule for the problem $(C, f, \bar d)$ defined over $\bar x$ (i.e. $\bar d_j = \min\{\d j, \sum_i \f ji \cdot \bar x_{i,j}$). 

\begin{heuristic}
    We define the set $A = \big\{ (\sum_j \floor{\x ij}, i) \big\}$ of multiplicity-block pairs. 
    We construct schedule $S_1$ by using the greedy multiset multicover algorithm given in \cite{rajagopalan+v:cover} on the instance $(A,C)$.
\label{alg:multisetmulticover_reduction}
\end{heuristic}
Step three of Algorithm~\ref{alg:log-approx} then constructs a schedule $S_2$ to satisfy any remaining demand given by the fractional components $\hat x$ via Algorithm~\ref{alg:greedy}, which greedily allocates the highest throughput machines until all demands are met.  Finally, step four of Algorithm~\ref{alg:log-approx} outputs the schedule $S$ such that:
$(a_1, \mu) \in S_1$ and $(a_2, \mu) \in S_2$ iff $(2(a_1+a_2),\mu) \in S$.

\begin{algorithm}
\setstretch{1.25}
    \KwInput{a \cms\ instance $(C,f,d)$ and block-job indexed variables $\hat x$}
    \KwInit{$\forall j, D_j \gets \min\{\d j, \sum_i \hat x_{i,j} \cdot \f ji\};\ S_2 \gets \varnothing$}
    \While{some job is not fully executed (i.e. $\sum_j D_j > 0$)}{
        $\mu \gets $ Algorithm~\ref{alg:throughput} on input $(C,f,D)$\;
        add $a^*$ machines $\mu$ to $S_2$, where $m^* =$ \hfill
        $\displaystyle \min_{j} \Big\{ a : D_j - \Big(a \cdot \sum_{i: \mu (i) = j} \f ji \Big) < \max_{i:\mu(i)=j} \Big\{ \min \{ \f ji, D_j\} \Big\} \Big\}$  \;
        $\forall j,\ D_j \gets \max \Big\{ 0, D_j - \Big( a^* \cdot \sum_{i: \mu(i) = j)} \f ij \Big) \Big\}$
    }
    \Return $S_2$
 \caption{Highest Throughput Placement First}
\label{alg:greedy}
\end{algorithm}

\begin{heuristic}
    \textit{On input $(C,f,d)$.} Iterate over each configuration $\sigma \in C$ and each block $i \in \sigma$. Assign to block $i$ the job $j$ that maximizes $\min\{\f ji, D_j\}$ where $D_j$ is the remaining demand of $j$. Output the maximum throughput machine.
\label{alg:throughput}
\end{heuristic}

In the remainder of the section, we provide analysis of Algorithm~\ref{alg:log-approx}.
Our first two lemmas establish that Algorithm~\ref{alg:log-approx} runs in polynomial time, and that an optimal solution to \constraints\ lower bounds the length of an optimal schedule. (See Appendix~\ref{sec:cms_appendix} for proofs.) 

\begin{lemma}
    Algorithm 1 runs in time polynomial in $n$, $k$, $|C|$, and $\max_{\sigma \in C} \{|\sigma|\}$.%
\label{lem:cms_polytime}
\end{lemma}

\begin{lemma}
    The optimal solution to \constraints\ is at most \opt.
\label{lem:LP}
\end{lemma}

The following lemmas establish bounds on the lengths of the schedules produced by Algorithms~\ref{alg:greedy} and \ref{alg:throughput}. (For a proof of Lemma~\ref{lem:throughput}, see Appendix~\ref{sec:cms_appendix}.)

\begin{lemma}
    The machine computed by Algorithm~\ref{alg:throughput} for an input problem instance has at least half the maximum throughput of any machine for that instance.  
\label{lem:throughput}
\end{lemma}

\begin{lemma}
    Given an instance $(C',f',d')$ with an optimal schedule of length $\rho$, Algorithm~\ref{alg:greedy} produces a schedule with length $3 \cdot \rho \cdot \log \sum_j d'_j$.
\label{lem:sum_of_demands}
\end{lemma}
\begin{proof}
    \newcommand{\da}[1]{d^{(a)}_{#1}}
    Let $S$ be the schedule produced by Algorithm~\ref{alg:greedy}. 
    We index machines $\mu_m$ by the order in which they are allocated by Algorithm~\ref{alg:greedy} (for the purposes of this proof, we treat machines individually, not as multiplicities). We define 
    $\da j = \max \Big \{0, \ \d j -  \sum_{m=1}^{a\rho} \ \sum_{i: \mu_m(i) = j} \f ji \Big\}$.
    Informally, $\da j$ is the amount of job $j$'s demand remaining after Algorithm~\ref{alg:greedy} schedules its first $a\rho$ machines.  Let $I_a = (C',f',\da{})$ be the instance defined over this remaining demand. We show that for any integer $a$, the total throughput of machines $a\rho+1$ through $a\rho + \rho$ of $S$ is at least $\frac{1}{4} \sum_j \da j$. This is sufficient to prove the lemma.

    Consider an arbitrary $a$ and the set $M$ of machines $a\rho + 1$ through $a\rho+\rho$ of $S$. Let $S^*$ be an optimal schedule for $I_a$, and let $\mu^*_m$ be the $m$th machine of $S^*$, ordered arbitrarily. (We can infer that the length of $S^*$ is at most $\rho$.)
    For every job $j$ and index $m$ (restricted to $a\rho + 1$ through $ (a+1)\rho$), we define 
    \begin{center}\begin{tabular}{c}
         $u_j = \min \Big\{ \da j,\ \sum_{m} \sum_{i: \mu_m(i) = j} \f ji \Big\}$  \\
         $v_{j,m} = \min \Big\{ \sum_{i:\mu_m(i)=j} \f ji,\ u_j - \sum_{m' < m} v_{j,m'} \Big\}$\\ 
         $v^*_{j,m} = \min \Big\{ \sum_{i:\mu^*_m(i)=j} \f ji,\ u_j - \sum_{m' < m} v^*_{j,m}  \Big\}$ \\
         $w^*_{j,m} = \min \Big\{\sum_{i:\mu^*_m(i)=j} \f ji - v^*_{j,m},\ \da j - \sum_{m'<m} w^*_{j,m} + v^*_{j,m} \Big\}$
    \end{tabular}\end{center}
    We also define $V_m = \sum_j v_{j,m}$ and $V^*_m = \sum_j v^*_{j,m}$ and $W^*_m = \sum_j w^*_{j,m}$.
    These definitions imply that $\sum_m V_m = \sum_j u_j$ and $\sum_m V^*_m = \sum_j u_j$ and $\sum_m W^*_m + V^*_m = \sum_j \da j$.
    In this way, $V_m$ represents the total reduction in demand when Algorithm~\ref{alg:greedy} allocates machine $\mu_m$, and $V^*_m$ (resp. $W^*_m$) represents the amount of demand satisfied by machine $\mu_m$ in $S^*$ that is (resp. not) satisfied by $S$. So it is sufficient to show that $\sum_m W^*m \le 2 \sum_m V_m$. Suppose, for the sake of contradiction, that for some machine $\mu_m$ we have $W^*_m > 2 \sum_m V_m$. Because $W^*_m$ represents demand not satisfied by $S$, Algorithm~\ref{alg:throughput} would choose $\mu^*_m$ rather that $\mu_m$, by Lemma~\ref{lem:throughput}. This is a contradiction, which proves the lemma. \qed
\end{proof}

\begin{theorem}
    Algorithm~\ref{alg:log-approx} is $O(\log (\max_{\sigma \in C}\{|\sigma|\} \cdot n \cdot k))$-approximate. 
\label{thm:cms}
\end{theorem}
\begin{proof}
    Let $S_1$ represent the schedule produced by Algorithm~\ref{alg:multisetmulticover_reduction} and let $S_2$ represent the schedule produced by Algorithm~\ref{alg:greedy}. 
    We first argue that $S_1$ has length $O(\log(\max_\sigma \{|\sigma|\} \cdot n))\cdot \opt$.
    Algorithm~\ref{alg:multisetmulticover_reduction} reduces scheduling the integer components of the variables to an instance of multi-set multi-cover in which there are $n$ elements and in which the largest covering multi-set has size $\max_\sigma \{|\sigma|\}$. The claim follows directly from Lemmas~\ref{lem:LP} and \ref{lem:multiset_multicover} (see Appendix~\ref{sec:cms_appendix}).
    
    We now show that $S_2$ has length $O(\log (\max_\sigma \{|\sigma|\} \cdot nk)) \cdot \opt$. Let $\hat d_j = \min\{\d j, \sum_i \f ji \cdot \hat x_{i,j}\}$ be the demands satisfied by $S_2$, and let $\hat f$ be the execution function scaled relative to $\hat d$. By Lemma~\ref{lem:sum_of_demands}, we need only to bound $\sum_j \hat d_j$.
    
    Let $S^*$ be the optimal schedule of $(C,\hat f,\hat d)$ and  let $\rho$ be the length of $S^*$.
    Since the optimal solution satisfies all demand, we have that 
    \[\textstyle\sum_j \hat d_j \le \sum_{\mu \in S^*} \sum_i \hat f_{\mu(i)}(i) \le \rho \cdot \max_{\sigma \in C}\{|\sigma|\} \cdot \max_{i,j} \hat f_j(i)  \]
    We can infer $\rho \le nk$ because $(C,\hat f, \hat d)$ is defined over $\hat x$, so each job can be completely executed by one block of each type. Also, the definition of $\hat x$ entails that for each $j$, every nonzero value of $\hat x_{i,j}$ (resp. $\hat f_j(i) \cdot \hat x_{i,j}$) is within a factor of $2k$ (resp. $k$) of every other.  After scaling, this implies $\max_{i,j} \{\hat f_j(i)\} \le 2k^2$. So, $\sum_j d'_j \le 2nk^3 \cdot  \max_{\sigma \in C}\{|\sigma|\}$ and $\log \sum_j d'j = O(\log( \max_{\sigma \in C}\{|\sigma|\} \cdot nk))$. 
    
    Finally, in defining $\hat x$, we rounded down $x^*_{i,j}$ if (i) $z^*_{i,j} < 1/2k$ or if (ii) $z^*_{i,j} \cdot \f ji < \max_{i'}\{z^*_{i',j} \cdot \f j{i'}\}/k$. Job $j$'s total reduction in demand from (i) is no more than $\d j \sum_i x^*_{i,j} - \bar x_{i,j} \le \d j/2$, which is accounted for by doubling $S_1$ and $S_2$ in the output. Job $j$'s total reduction in demand due to (ii) is at most $\max_{i'}\{z^*_{i',j} \cdot \f j{i'}\}$ which is accounted for in setting $\hat x_{i,j} = 2z^*_{i,j}$ for all remaining $i$'s. Each increases our approximation ratio by factors of two. \qed
\end{proof}


\junk{
The following corollary establishes a similar result for a variation of \numcms\ that allows for a restricted configuration set. The corollary follows from the fact that, in this setting, the maximum size of any configuration is $k$. 

\begin{corollary}
    In the setting where $\sum_{i \in \sigma} i$ is at most some input $k$ for every configuration $\sigma \in C$, Algorithm~\ref{alg:log-approx} is $O(\log nk)$-approximate.
\end{corollary}
}

%% file: LP_pseudo.tex
\section{\cms\ with $O(1)$ configurations}
\label{sec:cms_constant}
We consider \cms\ with $n$ jobs and a set $C$ of $O(1)$ configurations, each of arbitrary size. We first observe (see Appendix~\ref{app:constant_C}) that the problem is NP-hard to approximate to within a factor of two. Our main result in this section is a polynomial time algorithm with cost the minimum of $(2+\epsilon)\opt + |C|$ and $(3+\varepsilon)\opt$, for arbitrary $\varepsilon > 0$, where $\opt$ is optimal cost.  Our algorithm, given in Algorithm~\ref{alg:GraphMatching}, guesses the number of machines of each configuration in an optimal solution, to within a factor of $1 + \eps$ (see lines~3-4), and then builds on the paradigm of~\cite{LenstraShmoysTardos} by carefully rounding an extreme-point optimal solution for a suitable instantiation of \lp\ (given in line~6).  Using extreme-point properties, we establish the following lemma, the proof of which is in Appendix~\ref{app:constant_C} and closely follows~\cite{LenstraShmoysTardos}. 







\begin{algorithm}[ht]
    \KwInput{A \cms\ instance $(C,f,d)$}
    $L \gets \{\,\lfloor(1+\varepsilon)^i\rfloor \,\mid\, 0 \leq i \leq \log_{1+\varepsilon} (\sum_j d_j)\,\}$\\
    $Sol \gets (0,\infty)$\\
    \ForEach{$C^* \in P(C)$, the powerset of $C$}{
        \ForEach{$(m_{\sigma_1}, ..., m_{\sigma_{|C^*|}}) \in L^{|C^*|}$}{
            $B^* \gets \{b\in B \ |\  \exists c \in C^*, b\in c\}$ is the block set of $C^*$\\
            Construct the following feasibility LP, $LP_{f}$:
            \vspace{-5pt}
            \setcounter{tempcounter}{\arabic{equation}}
            \setcounter{equation}{1}
            \begin{align}
                &\sum_j \x ij \le \sum_{\sigma_s \in C^*} m_{\sigma_s} \cdot a_{\sigma_s, i} &&\forall \ \text{block types}\ i \in B^*  \label{LP':blocks} \tag{$1'$} \\
                &\sum_i \f ji \cdot \x ij \ge \d j &&\forall \ \text{jobs}\ j \\
                &\x ij \ge 0 &&\forall \ \text{block types}\  i \in B^* \text{, jobs}\ j 
            \end{align}
            \setcounter{equation}{\arabic{tempcounter}}\\
            \vspace{-5pt}
            \If{{\upshape $LP_{f}$ is feasible with extreme-point $x$}}{
                Graph $G \gets (J \cup B^*, E)$
                with $E = \{\, (j,b) \,\mid\, x_{b,j} > 0 \,\}$\\

                \ForEach{{\upshape Component $S \in G$ that has a cycle $K$}}{
                        Pick job $j$ in $K$, and let $b_1, b_2$ be its neighbors in $K$\\
                        {\bf if} {$x_{b_1,j}\cdot \f j {b_1} \geq x_{b_2, j}\cdot \f j {b_2}$} {\bf then} {$E \gets E \setminus (j, b_2)$}\\
                        {\bf else} {$E \gets E \setminus (j, b_1)$}\\
                        Make $j$ the root of the remaining tree $S$\\
                }
                \ForEach{{\upshape Job} $j' \in G\cap J$}{
                        \textbf{for} $p$, the parent of $j'$, \textbf{do} $x^*_{p,j'} \gets \lfloor 2x_{p,j'}\rfloor$\\
                        \textbf{foreach} Child $c_i$ of $j'$ \textbf{do} $x^*_{c_i,j'} \gets \lceil 2x_{c_i,j'}\rceil$\\
                    }
                \textbf{foreach} Configuration $\sigma_i \in C^*$ \textbf{do} $y^*_{\sigma_i} \gets 2m_{\sigma_i} + 1$\\
                {\bf if } {$\sum_{i \in y^*} i < \sum_{j \in Sol[1]} j$} {\bf then }{$Sol \gets (x^*, y^*)$}\\
                \Break
            }
        }
    }
    \Return $Sol$
\caption{Schedule for \cms\ with $O(1)$ configurations}
\label{alg:GraphMatching}
\end{algorithm}

\begin{lemma}\label{lem:pseudo_forest}
    Every component in graph $G$ of line~8 has at most one cycle.
\end{lemma}

\begin{lemma}
    \label{lem:feasible}
Algorithm~\ref{alg:GraphMatching} returns a feasible integer solution to \lp.
\end{lemma}

\begin{proof}
    Since the algorithm returns the least cost rounded solution over all iterations, we need to show that $(x^*, y^*)$ is a feasible integer solution to \lp. By definition $x^*_{i,j}$ and $y^*_\sigma$ are integers for each $i$, $j$, $\sigma$. 
    It remains to show that $(x^*, y^*)$ is feasible in \lp. Constraints \ref{LP:x-nonneg} and \ref{LP:y-nonneg} are true by definition of $x^*, y^*$. 
    
    We now consider constraint \ref{LP:blocks}. If a block type $b \notin B^*$, then this constraint is satisfied because $x_{b,j} = 0$ for all $j$, and thus $x^*_{b,j} = 0$ for all $j$. Now we consider blocks that are in $B^*$. By Lemma \ref{lem:pseudo_forest} we know that each component of $G$ has at most one cycle. In the algorithm we remove an edge from each of these cycles, so the resulting graph is a forest. Thus each block type $i$ has one parent and so is a child of one job. This means that all $x_{i,j}$ variables associated with $i$ are rounded as $\lfloor 2x_{i,j}\rfloor$, except for the parent of $i$, $p_i$. So we obtain
    \begin{align*}
        \sum_j x^*_{i,j} = \sum_{j\neq p_i}\lfloor 2x_{i,j}\rfloor + \lceil 2x_{i, p_i}\rceil &\leq 2\sum_j x_{i,j} + 1
        \leq 2\sum_{\sigma_s \in C^*} m_{\sigma_s} \cdot a_{\sigma_s, i} + 1\\ 
        &\leq \sum_{\sigma_s \in C^*} (2m_{\sigma_s}+1) \cdot a_{\sigma_s, i} 
        \leq \sum_{\sigma \in C} y^*_{\sigma}\cdot a_{\sigma, i},
    \end{align*}
    where the second inequality follows from constraint \ref{LP':blocks} since $x$ is a feasible solution to $LP_f$, and the third inequality holds since $i \in B^*$ implying that there is at least one $\sigma_s \in C^*$ such that $a_{\sigma_s, i} > 0$. Thus, constraint \ref{LP:blocks} is satisfied.

    Finally we consider constraint \ref{LP:demand}. First we consider some job $j$ that is not a job whose edge was removed in the cycle. Then, since $G$ becomes a forest after pruning edges we obtain that either the children or the parent of $j$ satisfy at least half of its demand.  If its children satisfy at least half of its demand then we have $\sum_{\text{children of $j$}} \f ji \cdot x_{i,j} \geq \frac{1}{2}d_j$ and thus we obtain
    \[
    \sum_i \f ji \cdot x^*_{i,j} \geq \sum_{\text{children of $j$}} \f ji \lceil 2 x_{i,j}\rceil \geq 2\sum_{\text{children of $j$}} \f ji \cdot x_{i,j} \geq d_j,
    \]
    so the constraint is satisfied. Otherwise, its parent $p_j$ satisfies at least half of its demand implying that $x_{p_j, j} \geq \frac{1}{2}$ since we have $\f j {p_j} \leq d_j$ by our assumption on the input. Then, $x^*_{p_j, j} = \lfloor 2x_{p_j, j}\rfloor > x_{p_j, j}$, yielding $\sum_i x^*_{i,j}\cdot \f ji \geq \sum_i x_{i,j}\cdot \f ji \geq d_j$ since $x$ is a feasible solution to $LP_f$. So the constraint is satisfied. 
    
    Finally we consider any job $j$ that had an edge removed in the cycle.  Assume without loss of generality that $(j,b_2)$ was removed from the graph. Since $j$ is the root of the tree it was in (by line 13), all of its neighboring blocks are its children. Then, we have
    \begin{eqnarray*}
    \sum_i x^*_{i,j}\cdot \f ij = \sum_{i\neq b_2} \lceil 2x_{i,j}\rceil \cdot \f ji \geq 2x_{b_1, j}\cdot \f {b_1}j + \sum_{i\neq b_1, b_2} 2x_{i,j} \cdot \f ji\\
    \geq x_{b_1, j}\cdot \f {b_1}j + x_{b_2, j}\cdot \f {b_2}j + \sum_{i\neq b_1, b_2} 2x_{i,j} \cdot \f ji \geq \sum_i x_{i,j}\cdot \f ij \geq d_j.
    \end{eqnarray*}
    The third to last inequality comes as a consqequence of line 11 and the fact $(j,b_2)$ was removed from the graph. So the constraint is satisfied in all cases. Thus $(x^*, y^*)$ is a feasible integer solution to \lp.
    \qed
\end{proof}

\begin{lemma}
    \label{lem:pseudo_rounding_poly}
    The runtime of Algorithm~\ref{alg:GraphMatching} is polynomial if $|C| = O(1)$.
\end{lemma}

\begin{theorem}\label{thm:graph_match_approx}
Algorithm~\ref{alg:GraphMatching} gives a $\min \{(3+\varepsilon)\opt,(2+\epsilon)\opt + |C|\}$ approximation in polynomial time if $|C| = O(1)$.
\end{theorem}

\begin{proof}
    Consider the iteration where $C^* = C^{\opt}$ where $C^{\opt}$ is the set of configurations used by an optimal integer solution. The algorithm will iterate through potential counts $m_\sigma$ for each $\sigma$ in $C^*$, round and return a schedule the first time $LP_{f}$ has a feasible solution; let $m_{\sigma_1}, ... m_{\sigma_{|C^*|}}$ be the $m$ values in this iteration.  By Lemma~\ref{lem:feasible}, the solution returned is feasible, and by Lemma~\ref{lem:pseudo_rounding_poly} the running time is polynomial.
    
    We now bound the cost by first arguing that $\sum_{\sigma_i} m_{\sigma_i} \leq (1+\varepsilon) \opt$.  Observe that the $y$ values in the optimal integer solution to \lp\ would yield a feasible solution to $LP_f$ if they equalled the corresponding $m$ values in $LP_f$ (namely by setting the $x$ variables in $LP_f$ to the $x$ values in the optimal integer solution to \lp). For each such $y_i$ value, consider $p_i$, the first power of $1+\varepsilon$ that is at least $y_i$. Then, we have $y_i \leq \lfloor p_i \rfloor \leq (1+\varepsilon)y_i$. Therefore, by definition of $L$, we will set values for the $m_{\sigma_i}$ such that they are greater than and within a factor of $(1+\varepsilon)$ of the $y$ values from the optimal integer solution. Thus they will be feasible, since they use at least as many of each configuration, and $\sum_{\sigma_i}m_{\sigma_i} \leq (1+\varepsilon)\opt$. Since we iterate through the $m$ values in increasing order of $\sum_{\sigma_i} m_{\sigma_i}$ we know that the first feasible solution will use at most this many configurations. 
    
    Now consider that the rounded solution $y^*$ has $ \sum_{\sigma} y^*_{\sigma} \le \sum_{\sigma_i} (2m_{\sigma_i} +1) = 2\sum_{\sigma_i} m_{\sigma_i} + |C^{OPT}| \leq 2(1+\varepsilon)\opt + |C^{OPT}|$. Since the optimal integer solution uses at least 1 of each configuration in $C^{OPT}$, we have that $\sum_{\sigma} y^*_{\sigma} \leq (3+\varepsilon)\opt$ and also that $\sum_{\sigma} y^*_{\sigma} \leq (2+\varepsilon)\opt + |C|$. \qed
\end{proof}

%% file: comb_const_conf_size.tex
\section{\cms\ with $O(1)$ configurations of $O(1)$ size}
\label{sec:ptas_constantk}
In this section, we consider \cms\ with $n$ jobs, a set $C$ of a fixed number of configurations with the additional constraint that each configuration has at most a constant number $k$ of blocks.  Let $b$ be the total number of block types.  Since $|C|$ and $k$ are both constant, $b \le k|C|$ is a constant.  In Appendix~\ref{app:constant_k_constant_size}, we present an optimal dynamic programming algorithm for the problem, which takes time $(nkd_{\max})^{O(b+|C|)}$; this is pseudo-polynomial time for constant $b$ and $|C|$.  In the following, we present our main result of this section, a PTAS for the problem.  
\smallskip

\noindent {\sl Blocks and patterns.}
We number the block types 1 through $b$ and we use \textit{$p$-block} to refer to a block of type $p$.  We partition jobs into two groups: the \textit{large} jobs $L$ and \textit{small} jobs $S$. A job $j$ is small if there exists a configuration $\sigma$ such that $f_j(\sigma) \ge \varepsilon d_j$; otherwise, $j$ is large.  (Here we use $f_j(\sigma)$ to denote the total demand satisfied if every block in configuration $\sigma$ is assigned to $j$.)

Let $\eps > 0$ be a given constant parameter, and let $\lambda = \eps/(2k)$.
We define a \textit{pattern} $\pi$ to be a size $b$ list of integers $\pi_1$ through $\pi_b$ that sum to no more than $k/\lambda^2$; $\pi_p$ denotes the number of $p$-blocks in pattern $\pi$. Let $W$ be the set of all possible patterns.  So, $|W| \le (k/\lambda^2)^b$.  We assign each small job a \textit{type}. Job $j$ is of type $t \in 2^W$ if each pattern $\pi \in t$  is such that the demand of $j$ is satisfied if $j$ is allocated $\pi_i$ $i$-blocks for $1 \le i \le b$.  So, the number of job types is at most $2^{(k/\lambda^2)^b}$.  Define constant $\gamma = 2^{(k/\lambda^2)^b}$.  
\smallskip

\newcommand{\LPeps}{PTAS-LP}
\noindent {\sl The linear program.}
We define a linear program \LPeps\ using the following notation.
In \LPeps, $\sigma$ ranges over all possible configurations in $C$, $p \in \{1, \ldots, b\}$ ranges over types of blocks, $x_{j,p}$ is the number of $p$-blocks dedicated to processing a large job $j$, $y_{\sigma}$ is the number of machines we use with configuration $\sigma$, $\sigma_p$ is the number of $p$-blocks in $\sigma$ (this is a constant), $z_{t,\pi}$ is the number of small jobs of type $t$ that are distributed according to pattern $\pi$, and $n_t$ is the number of small jobs of type $t$.  Recall that $\pi_p$ is the $p$th entry of $\pi$.
\LPeps\ minimizes $\sum_{\sigma \in C} y_\sigma $ subject to the following constraints
\begin{align}
    &\textstyle\sum_{j \in L} x_{j,p} + \sum_{t \in 2^W} \ \sum_{\pi \in W} (z_{t,\pi} \cdot \pi_p) \le \sum_{\sigma} y_\sigma \cdot \sigma_p &&\forall p \in [b]  \label{LPeps:servers} \\
    &\textstyle\sum_{p \in [b]} f_j(p) \cdot x_{j,p} \ge d_j &&\forall j \in L \label{LPeps:execution-large} \\
    &\textstyle\sum_\pi z_{t,\pi} \ge n_t &&\forall t \in 2^W \label{LPeps:execution-small}\\
    &x_{j,p} \ge 0 &&\forall j \in L, p \in [b] \label{LPeps:x-nonneg} \\
    &y_{\sigma} \ge 0 &&\forall \sigma \label{LPeps:y-nonneg} \\
    &z_{t,\pi}  \ge 0 &&\forall t \in 2^W, \pi \in W \label{LPeps:z-nonneg}
\end{align}

\noindent\textbf{Constraints.}
Constraint~\ref{LPeps:servers} guarantees that the total number of blocks of type $p$ that are used to execute jobs is no omre than the total number of available blocks of type $p$. Constraint~\ref{LPeps:execution-large} guarantees that each large job is fully executed, and constraint~\ref{LPeps:execution-small} guarantees that each small job is fully executed.  Constraints~\ref{LPeps:x-nonneg} through~\ref{LPeps:z-nonneg} are non-negativity constraints.

\medskip
\noindent Lemma~\ref{lem:small_job} establishes that it is sufficient to consider schedules in which small jobs are executed by a bounded number of blocks. Lemma~\ref{lem:lpeps} shows that \LPeps\ is a valid relaxation for the problem.  We defer the proofs to Appendix~\ref{app:constant_k_constant_size}.

\begin{lemma}
\label{lem:small_job}
For any schedule with $m$ machines, there exists a schedule with $m(1+k \lambda)$ machines in which each small job is executed by at most $k/\lambda^2$ blocks.
\end{lemma}

\begin{lemma}
\label{lem:lpeps}
    The value of \LPeps\ is at most $(1+ k \lambda)\opt$.
\end{lemma}

\begin{algorithm}[htb]
    \KwInput{$(C,f,d)$}
    Solve \LPeps; let $(x, y, z)$ be the solution computed.\;
    \If{$n \le k(|C| + \gamma)/\lambda$}{Compute and return an optimal solution using enumeration}
    \ForEach{large job $j$ and block type $p$}{
    $\widehat{x}_{j,p} = \lceil x_{j,p} \rceil$; 
    Assign $\lceil x_{j,p} \rceil$ blocks of type $p$ to job $j$\;
    }
    \ForEach{job type $t$ and pattern $\pi$}{
    Assign blocks per pattern $\pi$ to each job in $\lceil z_{t,\pi} \rceil$ small jobs of type $t$
    }
    \ForEach{configuration $\sigma$}{
    Use $\lceil y_{\sigma} \rceil$ machines with configuration $\sigma$\;
    }
\caption{Schedule for $O(1)$ configurations of $O(1)$ size}
\label{alg:PTAS}
\end{algorithm}

\begin{theorem}
Algorithm~\ref{alg:PTAS} computes a $(1 + \eps)$-approximation in polynomial time.
\end{theorem}
\begin{proof}
First, if $n \le k(|C|+\gamma)/\lambda$, then the algorithm returns an optimal solution.  Otherwise, since each machine has at most $k$ blocks, we obtain that $\opt \ge (|C|+\gamma)/\lambda$.  We will show that the number of machines used is at most $(1 + k\lambda)\opt + \lambda^2 k \opt + |C| + \gamma$, which is at most $(1 + 2k\lambda)\opt = (1 + \eps)\opt$.

Rounding up the $x$ variables increases the number of blocks by at most the number of large jobs times the number of block types.  Since each large job requires at least $1/\lambda^2$ machines, this increase in the number of blocks is at most $\lambda^2 k \opt$.  Rounding up the $z$ variables adds at most $1/\lambda^2$ blocks per small job type assigned to a given pattern.  This increases the number of blocks by at most $\gamma$.  Rounding up the $y$ variables increases the number of machines by $|C|$.  Taken together with the above increase in the number of blocks, each of which requires at most one machine, we find that the total increase is bounded by $\lambda^2 k \opt + \gamma + |C|$.  By Lemma~\ref{lem:lpeps}, the LP optimal is at most $(1 + k \lambda) \opt$, yielding the desired claim.

The linear program \LPeps\ has at most $nb + |C| + \gamma \log_\gamma$ variables and $b + n + \gamma$ linear constraints (other than the non-negativity ones), and can be solved in polynomial time.  The enumeration for $n \le k(|C|+\gamma)/\lambda$ is constant time, while the rest of the algorithm is linear in the number of variables.  The hidden constant, however, is doubly exponential in the number of configurations $|C|$ and the configuration size bound $k$, and exponential in $1/\eps$. 
\qed
\end{proof}

\junk{
\noindent Thus, if we have a $(1+k\varepsilon)OPT + k$ approximation (the exact additive factor might be different but still constant) then we can run the above first with $X = k\frac{1}{\varepsilon}$. If it finds a solution then return that, and that is exact. If not, then we know $OPT > \frac{k}{\varepsilon}$. Then our algorithm will return a solution that is $(1+k\varepsilon)OPT + \varepsilon OPT = (1+(k+1)\varepsilon)OPT$ approximate. This then can be made $(1+\varepsilon)OPT$ by scaling $\varepsilon$ by $\frac{1}{k+1}$, which is allowed since $k$ is constant.}

\junk{
\subsection{Optimal algorithm for constant number of job types}
\paragraph{Model.}
The number of machines $m$ is unbounded. The number of available blocks sizes $L$ is constant. The number $d$ of different demands and demand functions is constant. For each job type $j$, there is also an associated multiplicty $a_j$ indicating how many jobs we are given of that type. The objective is to pack all jobs into the smallest number of machines possible. 

\begin{theorem}
    There is an optimal polynomial time algorithm for this problem. 
\end{theorem}
We define 
\[  P = \{ (x,1) \in \mathbb{R}^{dL+1}_{\ge 0} : \sum_{i,j} i \cdot x_{i,j} \le k\} \text{and}\]
\[Q = \{ (x,b) \in \mathbb{R}^{dL+1}_{\ge 0} : \forall j,\ \sum_i i \cdot x_{i,j} \ge d_j \cdot a_j;\ 0 \le b \le B\}.  \]
Here $x_{i,j}$ is the number of $i$-blocks which execute a job of type $j$ on the current machine.
The theorem them follows from Theorem 2 in \cite{GoemansRothvoss.BinPacking.13}.
}

%% file: appendix.tex
\section{General \cms}
\label{sec:cms_appendix}

\begin{proof}[Proof of Lemma~\ref{lem:cms_hardness}]
Consider an arbitrary instance $I$ of multiset multicover. Let ${\cal U}$ denote the set of elements and ${\cal C}$ the collection of multisets in the set cover instance.  Let $r_e$ denote the coverage requirement for element $e$.  We can assume without loss of generality that there do not exist two multisets $S_1$ and $S_2$ with $S_1 \subseteq S_2$, since we can eliminate $S_1$ from the set collection otherwise.  We construct an instance of \cms\ where each multiset $S$ is a configuration and each element $e$ is both a block type and a job.  The job $e$ has demand $r_e$, which can only be satisfied by $r_e$ blocks of type $e$.  

Any multiset multicover solution, given by a collection $M$ of multisets, corresponds to a solution for \cms: each multiset $S$ in $M$ is a machine configured according to $S$. Therefore, the number of multisets in $M$ is the same as the number of machines in the \cms\ solution. 
Furthermore, since each element $e$ is covered $r_e$ times in $M$, it follows that each job $e$ has $r_e$ occurrences of block type $e$ included in \cms\ solution, thus satisfying the demand for $e$.  Similarly, every \cms\ solution with $m$ machines is a collection of $m$ multisets, with each multiset corresponding to the configuration of a machine.  Since the objective function value achieved by each of the two solutions is identical, the reduction is approximation-preserving.   
\qed
\end{proof}

The multiset multicover problem is as hard as set cover, which is NP-hard to approximate to within a factor of $(1 - \eps) \ln n$ for every $\eps > 0$~\cite{dinur+s:repetition}, where $n$ is the number of element. 
We thus obtain the same hardness for \cms\ where $n$ is the number of jobs.

\begin{proof}[Proof of Lemma~\ref{lem:cms_polytime}]
    Constraints \constraints\ consist of $n \cdot k \cdot |C|$ variables and $k + n + nk + |C|$ inequalities, and so can be solved in polynomial time. The polynomial runtime of Algorithm~\ref{alg:multisetmulticover_reduction} follows from \cite{rajagopalan+v:cover}, and the fact that our reduction to multi-set multi-cover is polynomial time.

    Algorithm~\ref{alg:throughput} executes in $O(|C|(n + k))$ time, so it remains to show only that the number of iterations in Algorithm~\ref{alg:greedy} is polynomial. Note that, in each iteration, there is some job $n$ and some block type $k$ such that the amount of $n$'s remaining demand that can be satisfied by scheduling a block of type $k$ is reduced by some amount. We also note that once this amount has been reduced, scheduling another block of type $k$ satisfies the remaining demand of $n$. So the maximum number of reductions is at most $2nk$. This proves the lemma.
    \qed
\end{proof}

\begin{proof}[Proof of Lemma~\ref{lem:LP}]
    Consider an arbitrary schedule $S$. We set $\y\sigma$ to the number of machines on which $S$ uses configuration $\sigma$. For each job $j$ and each block $i$, we set $\x ij$ to the number of blocks of type $i$ on which $S$ executes job $j$. Constraints \constraints\ follow straightforwardly from this assignment.
    \qed
\end{proof} 


\begin{proof}[Proof of Lemma~\ref{lem:throughput}]
    Let $\mu$ be the machine returned by Algorithm~\ref{alg:throughput} and let $\sigma$ be the configuration used by $\mu$. We show that the maximum throughput machine $\mu^*$ over $\sigma$ has throughput no more than twice that of $\mu$.

    We order the blocks of $\mu$ by the order in which Algorithm~\ref{alg:throughput} allocates them. For each job $j$, and each block type $i$, we define
    \begin{center}\begin{tabular}{c}
        $u_j = \min\{\sum_{i:\mu(i)=j} \f ji, \d j\}$ \quad and \quad $u^*_j = \min\{\sum_{i:\mu^*(i)=j} \f ji, \d j\}$  \\ \\[-0.5em]
        $v_i = \min\{\f {\mu(i)}i,\ u_{\mu(i)} - \sum_{i< i: \mu(i') = \mu(i)} v_{i'}\}$  \\ \\[-0.5em]
        $v^*_i = \min\{\f {\mu^*(i)}i,\ u_{\mu^*(i)} - \sum_{i< i: \mu^*(i') = \mu^*(i)} v^*_{i'}\}$ \\ \\[-0.5em]
        $w^*_i = \min\{\f {\mu^*(i)}i - v^*_i,\ u^*_{\mu^*(i)} - \sum_{i< i: \mu^*(i') = \mu^*(i)} w^*_{i'} + v^*_{i'}\}$
    \end{tabular}\end{center}
    These entail that $\sum_{i} v_i = \sum_j u_j$, and $\sum_{i} v^*_i \le \sum_j u_j$, and $\sum_{i} w^*_i + v^*_i = \sum_j u^*_j$. Informally, $v_i$ represents the increase in total throughput when Algorithm~\ref{alg:throughput} allocates block $i$, and $v^*_i$ (resp. $w^*_i)$ represents the throughput on $i$ of $\mu^*$ that is (resp. not) satisfied by $\mu$.   
    
    Since $\sum_i v^*_i \le \sum_j u_j$, it is sufficient to show that $\sum_i w^*_i \le \sum_i v_i$.  Suppose that, for some $i$,  $w_i > v_i$. Since $w_i$ represents demand not satisfied by $\mu$, and since Algorithm~\ref{alg:throughput} greedily chooses the block with the highest throughput, Algorithm~\ref{alg:throughput} would have assigned job $\mu^*(i)$ to block $i$ instead of job $\mu(i)$.  This yields a contradiction, which proves the lemma. 
    \qed
\end{proof}

The following lemma from Rajagopalan and Vazirani \cite{rajagopalan+v:cover} provides an approximation guarantee for multi-set multi-cover.

\begin{lemma}[Theorem 5.1 in \cite{rajagopalan+v:cover}]
    An instance of multi-set multi-cover consists of a universe $U$ of multiplicity-element pairs $(a,i)$ and a collection $S$ of multi-sets of elements $i$.
    The objective is to cover the whole multiplicity of elements with the minimum number of multi-sets.
    There exists a polynomial time greedy algorithm for multi-set multi-cover with approximation ratio $\log(|U| \cdot \max_{S' \in S} |S'|)$. 
\label{lem:multiset_multicover}
\end{lemma}

We provide further analysis of Algorithm~\ref{alg:greedy}, which could be applied on its own to achieve an $O(\log \sum_j \d j)$ approximation. 
The following lemma shows that our analysis of Algorithm~\ref{alg:greedy} is tight. 

\begin{lemma}
    There exist a family of instances with $n$ jobs, $k$ block types, and configuration set $C$ such that, when applied on its own, Algorithm~\ref{alg:greedy} produces a schedule of length $\Omega(\log \sum_j d_j)$ and $\Omega(\sqrt k)$ and $\Omega(n)$.
\end{lemma}
\begin{proof}
Define $k$ and $C$ for a given number of jobs $n$. 
Set $k = n+1$. 
There are two allowed configurations: $\set{k}$ which has one block of type $k$ and $\{1, 2, 3, \ldots, n\}$ which has $n$ blocks of types 1 through $n$. Jobs are indexed 1 through $n$. The demand of job $j_\ell$ is $\d {j_\ell} = 2^\ell$. We define $\f {j_\ell}i = 0$ when $i \ne \ell,k$, and $\f {j_\ell}{\ell} = 2^{\ell-1}$, and $\f {j_\ell}{k} = 2^\ell$.

\textbf{Opt.} Executes all jobs on two machines using configuration $[1,2,\ldots,n]$. 

\textbf{Alg.} Executes all jobs on $n$ machines using configuration $[k]$. 

So the approximation ratio for this family of instances is a factor of $n \approx \sqrt k \approx \log \sum_j d_j$.
\qed
\end{proof}

\section{\cms\ with a fixed number of configurations}
\label{app:constant_C}
\begin{lemma}
    \label{lem:cms_fixed_configs_hardness}
    \cms\ with a fixed number of configurations is hard to approximate to within a factor of 2.
\end{lemma}
\begin{proof}
We present a reduction from Partition to combinatorial \cms.  Given an instance of Partition with a set $S$ of $n$ elements $0 < a_1 < a_2 < \cdots < a_n$, we construct the following instance. We consider one configuration that contains $n$ blocks all of a different type, labeled $1,...,n$. We have two jobs $j_1, j_2$ both with the same demand table given by $f(i) = a_i$. The demand for each job is $\frac{1}{2}\sum_i a_i$.
\junk{We present a reduction from Partition to \cms\ with one configuration.  Given an instance of Partition with a set $S$ of $n$ elements $0 < a_1 < a_2 < \cdots < a_n$, we construct the following instance of \cms.  We have two jobs $j$ with a demand table given by a step function with $n$ steps: for $0 \le i \le n-1$, $x \in [a_i, a_{i+1})$, $f_j(x) = a_i$ ($a_0 = 0$).  There is exactly one configuration consisting of $n$ blocks $a_i$, $1 \le i \le n$.  The capacity of each machine is $\sum_i a_i$ and the demand for each job is $\sum_i a_i/2$.} 

We claim that the number of machines needed for scheduling the job is one if and only if the Partition instance has a yes answer.  If the Partition instance has a yes answer, then there exists a way to split the $n$ blocks into two parts so that part's value adds up to $k$.  We use one machine, and assign the blocks to each job according to the Partition solution.  The demand table ensures that the demand of the job is satisfied.  If the demand of the two jobs is satisfied by one machines, then the machine serves a total demand of $\sum_i a_i$.  By the demand table, each block satisfies a demand of $a_i$ for some $i$, implying the existence of a two parts of items from $S$, each part's total size adding up to $\sum_i a_i/2$.
\qed


    
\end{proof}

\begin{lemma}
    The number of nonzero variables in $x$ of line~7 is at most $n + |B^*|$.
    \label{lem:nonzero_vars}
\end{lemma}
\begin{proof}
    Using extreme point properties we know that the number of tight constraints is at least as many as the number of variables. This leaves only $n+|B^*|$ constraints to not be tight (coming from constraints~\ref{LP':blocks} and \ref{LP:demand}).
    \qed
\end{proof}

\begin{proof} [Proof of Lemma~\ref{lem:pseudo_forest}] 
    This proof follows a similar structure as the proof of Lemma 17.6 in \cite{Vazirani}. We will use a proof by contradiction. First, consider a component in $G$, called $G_c$. Then consider the restriction of the LP, $LP_c$, to only the jobs and block types present in the component. Also let $x_c$ be the restriction of $x$ to those jobs and blocks present in the component. Let $x_{\bar c}$ be the rest of $x$. Note that $x_c$ is a feasible solution to $LP_c$ since all the blocks that satisy demand for jobs in $LP_c$ are connected to the jobs in $G$ and thus are included in $LP_c$, so we continue to satisfy all the demand for these jobs. Now assume for contradiction that $x_c$ is not an extreme point in $LP_c$. Then $\exists x_1, x_2, \lambda$ where $x_1$ and $x_2$ are feasible solutions to $LP_c$ and $\lambda \in (0,1)$ such that we have $x_c = \lambda \cdot x_1 + (1-\lambda)\cdot x_2$. 
    
    Now we show that $x_1 + x_{\bar c}$ and $x_2 + x_{\bar c}$ are feasible solutions to $LP$. First consider that $x_1, x_2$ have disjoint jobs and block types from $x_{\bar c}$. Thus, we can consider the constraints separately. Furthermore, together they cover all the constraints (since they cover all jobs and block types). Thus we need only verify that $x_1,x_2$ satisfy their constraints, and $x_{\bar c}$ satisfies its constraints. Since $x_1, x_2$ are feasible solutions to $LP_c$ we know they satisfy the constraints in $LP$ relevant to them. And since $x_{\bar c}$ is part of the feasible solution $x$, it must also satisfy the contraints relevant to it. Between the two, all the constraints of the $LP$ are satisfied, since together they cover all jobs and blocks.

    But then since $x = \lambda \cdot (x_1 + x_{\bar c}) + (1-\lambda)\cdot (x_2 + x_{\bar c})$ we can say that $x$ is a convex combination of two other solutions. Thus, $x$ is not an extreme point solution. But, since $x$ is an optimal solution to the $LP$, it must also be an extreme point solution. Thus we reach a contradiction.

    Therefore, $x_c$ must be an extreme point solution in $LP_c$. But then, by Lemma~\ref{lem:nonzero_vars} we have that the number of edges in $G_c$ must be at most the number of jobs and blocks in $G_c$. In other words, the number of edges is at most the number of nodes. Therefore, $G_c$ is a pseudo-tree, and $G$ is a pseudo-forest.
    \qed
\end{proof}

\begin{proof}[Proof of Lemma~\ref{lem:pseudo_rounding_poly}]
    The first for loop in the algorithm ranges over $2^{|C|}$ values. The inner for loop ranges over $(\log L)^{|C^*|}$ values. Remember that $L = \sum_j d_j$. But then $L \leq n\cdot \max_j d_j$. Thus the inner loop ranges over $\leq (\log (n\cdot \max_j d_j))^{|C^*|} \leq (\log n + \log \max_j d_j)^{|C|}$ values. Since $d_j$ is specified as a number, it is specified using $\log d_j$ bits. Thus the inner loop runs a number of times polynomial in the input, except for the number of configurations. Lastly we analyze the body of the inner for loop. The size of the LP is polynomial in the size of the input, and thus constructing and solving it takes time polynomial in the size of the input. Constructing the graph takes time polynomial in the size of the LP, as does rounding using the graph. Thus overall the runtime of the algorithm is polynomial in the size of the input, except for it being exponential in the number of configurations.
    \qed
\end{proof}

\section{\cms\ for $O(1)$ configurations of $O(1)$ size}
\label{app:constant_k_constant_size}
{\bf A pseudo-polynomial time algorithm.} 
We present an optimal algorithm, based on dynamic programming, that takes time polynomial in $n$ and the maximum demand.   Recall that $C$ denotes the set of configurations, and $|C|$ is constant.  Let $N$ denote the total number of machines available.  Then, there are $\binom{N + |C| - 1}{|C| - 1}$ different ways of distributing the $N$ machines among these configurations.  Each way yields a specific number of blocks of each type.  For given $n_i$, $1 \le i \le b$, let $S(j, n_1, n_2, \ldots, n_{b})$ be True if the demand of jobs 1 through $j$ can be satisfied using $n_i$ blocks of type $i$th, for each $i$.  Then, we have
\begin{eqnarray*}
S(j, n_1, n_2, \ldots, n_{b}) = \bigvee_{m_i \le n_i, \forall i} & \left(S(j-1, n_1 - m_1, n_2 - m_2, \ldots, n_{b} - m_{b})\right. \\
& \bigwedge \left.T(j, m_1, m_2, \ldots, m_{b})\right),
\end{eqnarray*}
where $T(j,m_1, m_2, \ldots, m_{b})$ is true if and only if the demand of $j$ can be satisfied using $m_i$ blocks of type $i$, for each $i$.  Note that $T(j,m_1, m_2, \ldots, m_{b})$ can be computed easily by inspecting the demand table of job $j$ and its demand $d_j$.  

The algorithm computes $S(j, n_1, n_2, \ldots, n_{b})$ for $1 \le j \le n$, $n_i \le Nk$; the number of different tuples equals $n(Nk)^b$.  The time taken to compute a given $S(j, n_1, n_2, \ldots, n_{b})$, given $S(j-1, n_1 - m_1, n_2 - m_2, \ldots, n_{b} - m_{b})$ for all choices of $m_i$'s, is proportional to the number of different choices of $m_i$'s, which is bounded by $\binom{N + |C| - 1}{|C| - 1}$.  We thus obtain that $S$ can be computed in $n(Nk)^{O(b+|C|)}$.  This computation, coupled with a binary search over possible values of $N$, yields the desired algorithm.  Since $N$ is bounded by $n$ times the maximum demand, we obtain a pseudopolynomial time optimal algorithm if $|C|$ and $b$ are bounded. 

\begin{proof}[Proof of Lemma~\ref{lem:small_job}]
    Consider any placement $P$ that uses $m$ machines. 
 Suppose a small job $j$ is in more than $k/\varepsilon^2$ blocks in $P$.  Since each configuration is of size at most $k$, it follows that the job is placed in at least $1/\lambda^2$ machines.  Since $j$ is small, there exists a configuration $\sigma$ such that $f_j(\sigma) \ge \lambda d_j$.  We remove job $j$ from each machine to which it is assigned in $P$ and place it in $1/\lambda$ additional machines, each with configuration $\sigma$, guaranteeing that the demand of $j$ is satisfied.  Since each machine can hold at most $k$ small jobs, this modification of $P$ results in the increase in the number of machines by a factor of at most $(1 + k\lambda)$, yielding the desired claim.
 \qed
\end{proof}

\begin{proof}[Proof of Lemma~\ref{lem:lpeps}]
    Let $A$ be an optimal placement of the jobs on $m$ machines.  Using Lemma~\ref{lem:small_job}, we first compute a new placement $B$ using at most $m(1 + k \lambda)$ machines in which each small job is placed in at most $1/\lambda^2$ machines.
    
    We now define variable assignments so that the value of \LPeps\ is no more than $(1 + k \lambda)m$. 
    For each large job $j$ and each block of size $p$, set $x_{j,p}$ to be the number of $p$-blocks on which $B$ executes $j$. For each small job type $t$ and each pattern $\pi$, set $z_{t,\pi}$ to be the number of small jobs that are executed in pattern $\pi$ according to $B$.  Note that since each small job is placed in at most $1/\lambda^2$ machines, and hence at most $k/\lambda^2$ blocks, the placement of each small job follows one of the patterns in $W$.     
    Set $y_\sigma$ equal to the number of machines with configuation $\sigma$ according to $A$.
    
    It is easy to see that constraints (\ref{LPeps:execution-large} - \ref{LPeps:z-nonneg}) are satisfied. To see that constraint~\ref{LPeps:servers} is satisfied, observe that each machine used by $B$ either has some block executing a large job (in which case it contributes toward the first term of~\ref{LPeps:servers}) or it has some block executing a small job (in which case it contributes toward the second term). Therefore, the left hand side of~\ref{LPeps:servers} counts the total number of blocks needed to complete all the jobs, while the right hand side computes the total number of blocks supplied by the machines. 
    \qed
\end{proof}